\documentclass{IEEEtran}
\pdfoutput=1
\usepackage{amsmath, amsthm, amsfonts, amssymb, amsbsy}

\usepackage{graphicx}
\usepackage{subfigure}
\usepackage{textcomp}
\usepackage{xcolor}

\usepackage{pstricks}
\usepackage{algorithm}
\usepackage{algorithmicx}

\usepackage{multirow}
\usepackage{bm}
\usepackage{cite}
\usepackage{url}

\usepackage{nopageno}

\usepackage{bbm}
\usepackage{dsfont}

\usepackage[english]{babel}
\usepackage{algpseudocode}

\newtheorem{theorem}{Theorem}

\newtheorem{definition}{Definition}

\newtheorem{lemma}{Lemma}

\newcommand{\twotriangle}{\hfill $\bigtriangleup \bigtriangleup$  }
\newcommand{\eax}{\twotriangle  \end{example}}
\newcommand\bim{\begin{itemize}}
\newcommand\eim{\end{itemize}}

\title{Downsampling and Transparent Coding for Blockchain
}
\author{Qin Huang$^1$, Li Quan$^1$, and Shengli Zhang$^2$\\
$^1$School of Electronic and Information Engineering,\\ Beihang University, Beijing, China, 100191\\
$^2$College of Information Engineering,\\ Shenzhen University, Shenzhen, China, 518060\\
\thanks{This work was supported by Young Elite Scientists Sponsorship Program by CAST under Grant 2017QNRC001. This article was presented in part at IEEE INFOCOM 2019~\cite{8845124} and ACM CoNEXT 2019~\cite{Quan2019Transparent}.  (Qin Huang and Li Quan contributed equally to this work.) (Corresponding author: Qin Huang and Shengli Zhang. Email: qhuang.smash@gmail.com, zsl@szu.edu.cn.)}}

\begin{document}


\maketitle
\thispagestyle{empty}
\begin{abstract}
With the development of blockchain, the huge history data limits the scalability of the blockchain. This paper proposes to downsample these data to reduce the storage overhead of nodes. These nodes keep good independency, if downsampling follows the entropy of blockchain. Moreover, it demonstrates that the entire blockchain history can be efficiently recovered through the cooperative decoding of a group of nodes like fountain codes, if reserved data over these nodes obey the soliton distribution. However, these data on nodes are uncoded (transparent). Thus, the proposed algorithm not only keeps decentralization and security, but also has good scalability in independency and recovery.
\end{abstract}


\section{Introduction}
Full nodes, which store the entire blockchain history, can not only serve its own transactions, but also support the recovery of other fail nodes. Therefore, most blockchain systems recommend their users, at least commercial users, to become full nodes. However, the storage cost of becoming a full node is pretty high. For instance, the Bitcoin blockchain size has exceeded $239GB$ until September 2019, and is still growing at a rate of about $50GB$ per year. Even worse, there is no reward or only a few rewards for maintaining a full node. Therefore, maintaining a full node is neither practical nor cost-effective for general users or even some commercial users. Even though there are more than $40$ million Bitcoin wallets around the world, it's not surprising that only about $9000$ of them are full nodes. In other words, only $0.0225\%$ of the nodes are responsible for most of the security guards and services for other nodes. It limits the scalability of the blockchain.

Recently, many researchers have devoted themselves to the storage challenge. In 2008, Nakamoto \cite{Nakamoto2008Bitcoin} proposed Simplified Payment Verification (SPV) and pruned nodes. In 2018, Leung et al. \cite{leung2018vault} proposed Vault to speed up bootstrapping for a new node. Dryja \cite{dryja2019utreexo:} proposed Utreexo to optimize the Bitcoin \emph{unspent transaction outputs} (UTXO) set. These remarkable studies have successfully reduced the storage requirements of nodes. However, they had to sacrifice the part of the blockchain history. The lack of entire blockchain history of many nodes will centralize the entire blockchain history on a small number of nodes, jeopardizing the decentralization of the blockchain. What's more, since historical records play an important role in application scenarios such as supply chain and copyright registration, the lack of blockchain history can also limit the application of blockchain. In addition to these studies, some researchers have tried to deal with this contradiction from the perspective of information theory and coding theory. In 2018, Perard et al. \cite{perard2018erasure} proposed to use erasure coding to create low storage blockchain nodes. In this way, the entire blockchain history can be recovered from a subset of nodes, but the encoded history cannot be used directly unless decoded. Therefore, it is important to find a low-cost blockchain history storage method that can balance decentralization and transparency (stored data can be used directly by the node without decoding) to deal with the contradiction between storage and scalability.

This paper reduces the storage overhead of nodes by downsampling history data. After downsampling, a node only needs to synchronize and maintain a small proportion of all the blocks. We demonstrate that our well-designed \emph{downsampled nodes} (DSNs) can provide good scalability in independency and recovery.

\begin{itemize}
  \item {\bf Independency:} We prove that DSNs are able to independently verify and broadcast transactions. Moreover, DSNs containing blocks with higher entropy have better verification accuracy than DSNs containing blocks with lower entropy. Therefore, we downsample blocks following the distribution of block entropy to achieve better verification accuracy. Furthermore, in the UTXOs model, because of the sequential feature of blockchain, it can be determined whether a transaction output is valid if all blocks after that transaction are known. Therefore, if containing continuous latest blocks, DSNs will not be deceived by a malicious node to believe a transaction that references an invalid input.

     \item {\bf Recovery:} The entire history data can be recovered through the cooperation of a group of DSNs. Usually erasure codes are used for data recovery in distributed storage systems. However, for blockchain, parity-checks of erasure codes can not be verified by the hash value of block headers, and are unavailable before decoding. Instead we propose that uncoded (transparent) transactions are stored on a set of DSNs following the soliton distribution. Only during the recovery, these DSNs encode their transactions to help fail nodes. Suppose that we want to recover a segment with $K$ transactions, and DSNs have $O(\ln (K / \epsilon))$ transactions of this segment on average. Each DSN encodes its related transactions into a codeword by simple bitwise sums. We demonstrate that the recovery probability is $1-\epsilon$ if $K+O\left(\sqrt{K} \ln ^{2}(K / \epsilon)\right)$ codewords of such DSNs are received.

\end{itemize}

This paper is organized as follows. Section II gives the background of the blockchain. Section III presents the downsampling nodes of blockchain. Section IV introduces the information entropy of blockchain to guide the downsampling. Section V achieves the recovery of history data of blockchain with transparent coding. Section VI concludes this paper.

\section{Background}

As a sequential, open and distributed ledger, blockchain cryptographically secures records of transactions \cite{narayanan2016bitcoin}. It is able to tolerate the failures of the Byzantine Generals' Problem~\cite{lamport1982the}. Bitcoin is the first and most typical blockchain system. This section uses it as an example to describe the structure and principle of blockchain.

\begin{figure}[htbp]
\centerline{\includegraphics[width=3.3in]{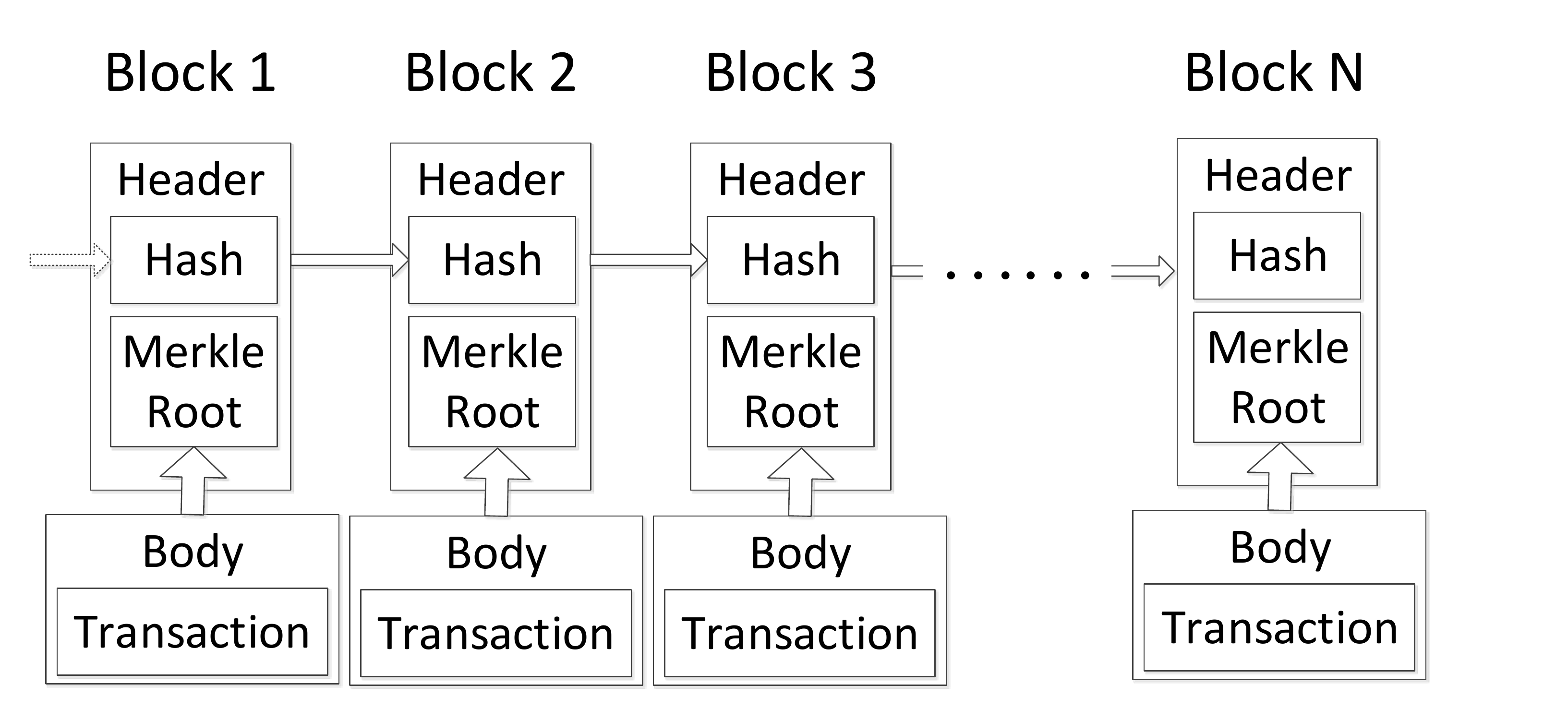}}
\caption{Structure of Bitcoin.}
\label{fg:blockchain_structure}
\end{figure}

The structure of Bitcoin is illustrated in Fig.  \ref{fg:blockchain_structure}. It consists of $N+1$ blocks sequentially stored from block $0$ to block $N$. Every block contains a header and a body. The header has the previous header's hash, current timestamp and Merkle root. With the hash identifier of the previous header, blocks are sequentially chained together by the hash pointer. The body contains transactions and their Merkle tree. This tree is the hash tree of transaction identifiers. The root of this tree is called the Merkle root, which is stored in the header.

A typical transaction has inputs and outputs
\begin{equation}
T \equiv \{T^{in},T^{out}\},
\end{equation}
where $T^{in}$ is the set of all inputs, and $T^{out}$ is the set of all outputs. A valid transaction should use UTXOs as its inputs, except coinbase transaction. When a node processes a new received transaction $T$, it needs to check $T$ with Algorithm \ref{alg:tveri} \cite{antonopoulos2014mastering}, where ${\cal S}_{processed}$ denotes the set of processed transactions, and ${\cal S}_{UTXO}$ denotes the set of UTXOs.

\begin{algorithm}[htb]
\caption{Transaction Verification\cite{antonopoulos2014mastering}}\label{alg:tveri}
\begin{algorithmic}[1]
    \Require
      $T$;
    \Ensure
      true or false;
    \If {$T \in {\cal S}_{processed}$} \Return false
    \EndIf
    \State //Check whether addresses in the $T$ are valid.
    \If {CHECKADDR ($T$) $=$ false} \Return false
    \EndIf
    \State //Check whether the originator of $T$ is the legal owner of the input address.
    \If {CHECKOWNER ($T$) $=$ false} \Return false
    \EndIf
    \If {$T \notin {\cal S}_{UTXO}$} \Return false
    \EndIf
    \If {$\sum T^{in} \ge \sum T^{out}$} \Return true
    \EndIf
\end{algorithmic}
\end{algorithm}

If Algorithm \ref{alg:tveri} returns true, $T$ will be regarded as a valid unconfirmed transaction, then it will be broadcasted. In order to check all new transactions safely, nodes need to download and verify the blocks from the genesis block to the latest block. Then, miners pack several valid transactions into a block. Then following the consensus mechanism \cite{Nakamoto2008Bitcoin} of blockchain, a new transaction will be carried out by all the nodes.

In fact, the blockchain is a state machine whose state transition is based on the transaction\cite{Gavin2019Ethereum}. For a blockchain, starting from a genesis state, the blockchain executes transactions one by one in order, and finally reaches a certain final state. The state of the blockchain could be defined as follows:
\begin{definition}
The state of the blockchain is a set of attributes that reflects the characteristics of the blockchain, which contains the time, consensus results, difficulty, and balances of different owners.
\end{definition}

The state transition caused by the execution of the transaction is
\begin{equation}
\boldsymbol{\sigma}_{t+1} \equiv \Upsilon\left(\boldsymbol{\sigma}_{t}, T_{t}\right),
\end{equation}
where $\Upsilon$ is the state transition function, $\boldsymbol{\sigma}$ is the state of the blockchain system, and $t$ is the time slot.

For a block, since it contains multiple transactions, the state transition of the block can be regarded as a continuous transaction state transfer,
\begin{equation}
\boldsymbol{\sigma}_{t+1} \equiv \Pi\left(\boldsymbol{\sigma}_{t}, B\right),
\end{equation}
\begin{equation}
B \equiv\left(\ldots,\left(T_{0}, T_{1}, \ldots\right), \ldots\right),
\end{equation}
\begin{equation}
\Pi(\boldsymbol{\sigma}, B) \equiv \Omega \left(B, \Upsilon\left( \Upsilon\left(\boldsymbol{\sigma}, T_{0}\right), T_{1}\right) \dots \right),
\end{equation}
where $B$ is the block, which contains a series of transactions and related parameters; $\Pi$ is the block-level state transition function; $\Omega$ is the block finalization function, mainly through the consensus algorithm to determine whether a certain state after the state transition is finalized and could be added to the blockchain.

As saving the entire blockchain, full nodes can check the security of all transactions. However, saving all blocks increases the amount of data the node needs to process. These data need to be downloaded, stored, verified, indexed and updated. This increases the nodes' storage overhead and reduces the efficiency of processing transactions.

%

\section{Downsampling Nodes of Blockchain}
As the volume of transactions rapidly increases, the storage of blockchain history may be unaffordable for most nodes. A node, e.g. SPV node, may discard all block bodies and only save block headers, but it will suffer scalability problems. Since it only stores the block headers of the best chain, it can neither independently verify and broadcast transactions, nor help recover history data of the blockchain.

\begin{algorithm}[htb]
\caption{Transparent Downsampling Blockchain Algorithm}\label{alg:tdba}
\begin{algorithmic}[1]
    \Require The number of reserved blocks $\delta$;
    \State // The DSN $P$ stores the set of entire block headers $\mathcal{H}$.
    \State $P$: STORE($\mathcal{H}$);
    \State //Get the reserved set ${\cal D}$, where ${\cal D}$ is the set of $\delta$ block bodies.
    \State ${\cal D} =$ GETRES($\mathcal{H}$);
    \State $P$: STORE(${\cal D}$);
\end{algorithmic}
\end{algorithm}

In this section, we propose a node to store all the block headers and partial block bodies, named \emph{downsampling node}. Each DSN can be generated by Algorithm \ref{alg:tdba} independently. DSNs are able to  independently verify and broadcast transactions and contribute to the recovery of blockchain history as follows:

\begin{itemize}
\item The DSN uses block bodies of ${\cal D}$ to generate UTXOs pool ${\cal P}_{UTXO}$. When a new $T$ is received, the DSN will check it with Algorithm \ref{alg:dtveri}. If all the checks are passed, $T$ will be regarded as a valid unconfirmed transaction and broadcast.
\begin{algorithm}[htb]
\caption{Transaction Verification of DSNs}\label{alg:dtveri}
\begin{algorithmic}[1]
    \Require
      $T$;
    \Ensure
      true or false;
    \If {$T \in {\cal S}_{processed}$} \Return false
    \EndIf
    \State //Check whether addresses in the $T$ are valid.
    \If {CHECKADDR ($T$) $=$ false} \Return false
    \EndIf
    \State //Check whether the originator of $T$ is the legal owner of the input address.
    \If {CHECKOWNER ($T$) $=$ false} \Return false
    \EndIf
    \If {$T^{in} \in {\cal P}_{UTXO}$} \Return false
    \EndIf
    \If {$\sum T^{in} \ge \sum T^{out}$} \Return true
    \EndIf
\end{algorithmic}
\end{algorithm}

\item The entire history data can be recovered through the cooperation of a group of DSNs. When a node $P'$ needs to recover a block $B'$, it can request the transaction via Algorithm \ref{alg:gb}. In addition, when the node $P'$ needs to recover all blocks, Algorithm \ref{alg:gb} can be used in parallel.

\begin{algorithm}[htb]
  \caption{ Recover a Block }
  \label{alg:gb}
  \begin{algorithmic}[1]
    \Require
      \Statex The set of neighbor nodes of node $P'$, ${\{P^1,P^2,\dots,P^{\cal{M}}\}}$;
      \Statex The block header of $B'$, $h'$;
    \Ensure
      The block $B'$;
    \For{i=1 to $\cal{M}$}
      \State //$P'$ sends $h$ to $P^i$.
      \State $P' \to P^i$: $h$;
      \If {$P^i \to P'$: $B'$} \Return $B'$
      \EndIf
    \EndFor
  \end{algorithmic}
\end{algorithm}

\end{itemize}

It is worth mentioning that DSNs are able to work independently. Moreover, DSNs only make slight changes in choosing the reserved data. Thus, it is possible to keep the original network architecture and consensus algorithms of the blockchain.

In the following two sections, we propose entropy-based downsampling and construct transparent coding to optimize scalability in independency and recovery, respectively.

\section{Entropy-Based Downsampling and Independency of DSNs}
In blockchain, the verification of a transaction strongly relies on the most recent state. It is important for DSNs to select the block bodies to get as many parts of the most recent state as possible. Therefore, we propose to downsample block bodies following the entropy of blockchain.

\subsection{Entropy-based downsampling}
\begin{figure}[htbp]
\centerline{\includegraphics[width=3.3in]{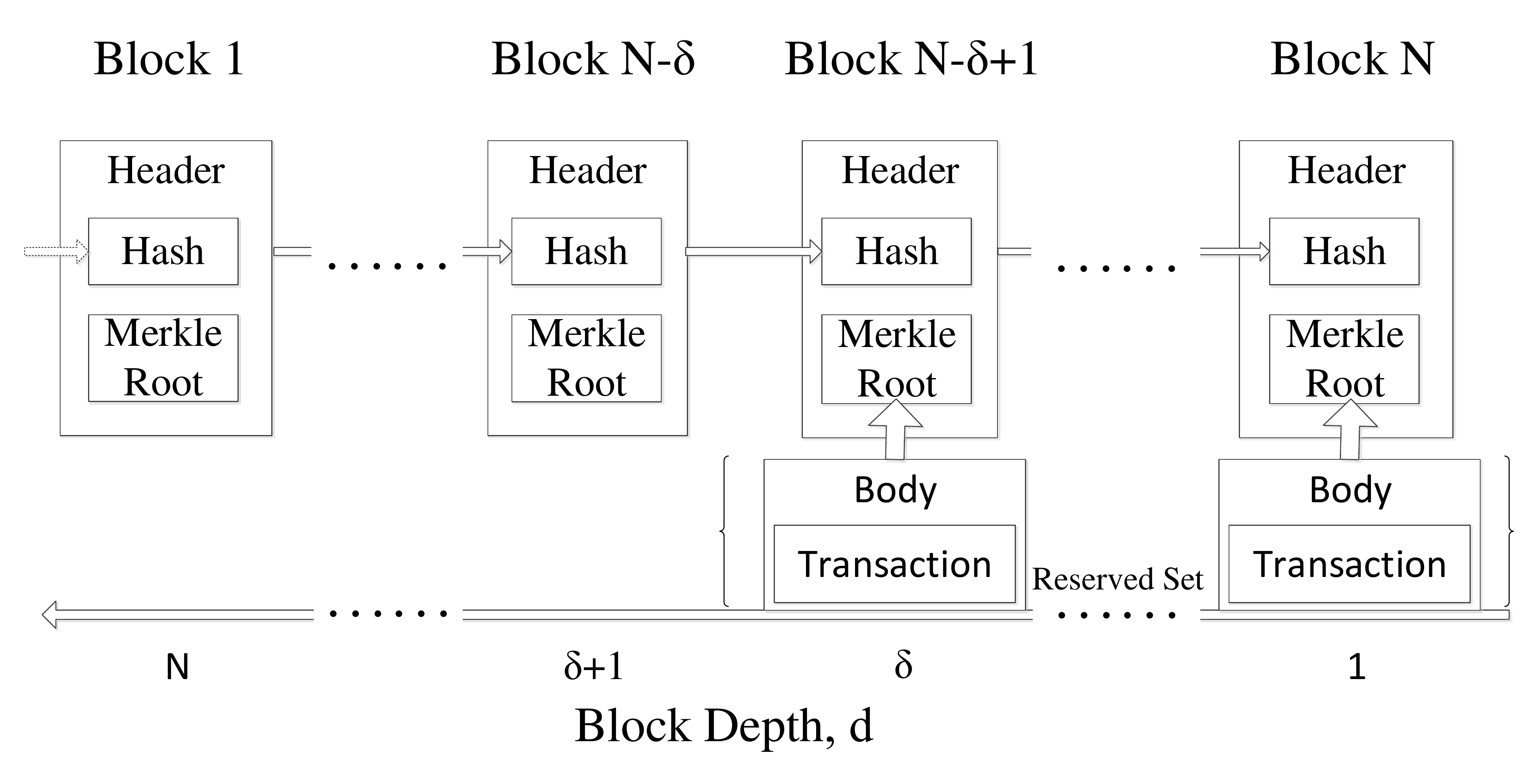}}
\caption{Structure of blockchain on DSNs.}
\label{fg:blockdepth}
\end{figure}
Suppose every block can be viewed independently. The information entropy of various blocks at a time slot is denoted by $H(d)$, where the positive integer $d$ is the block depth as illustrated in Fig. \ref{fg:blockdepth}. The latest block has depth of $1$. Then, the blockchain can be downsampled by the following Algorithm \ref{alg:dba}.

\begin{algorithm}[htb]
\caption{Entropy Based Downsampling Blockchain Algorithm}\label{alg:dba}
\begin{algorithmic}[1]
    \Require
    \Statex The information entropy of blocks, $H(d)$;
    \Statex The number of reserved blocks, $\delta$;
    \State // The DSN $P$ stores the set of entire block headers $\mathcal{H}$.
    \State $P$: STORE($\mathcal{H}$);
    \State //Get the $\delta$ block headers whose block bodies have the largest entropy.
    \State ${\cal H}' =$ MAXN\_H ( $\mathcal{H}, H(d), \delta$ );
    \State //Get the reserved set ${\cal D}$, where ${\cal D}$ is the set of $\delta$ block bodies.
    \State ${\cal D} =$ GETRES($\mathcal{H}'$);
    \State $P$: STORE(${\cal D}$);
\end{algorithmic}
\end{algorithm}
In Algorithm \ref{alg:dba}, the number $\delta$  of reserved block bodies could be set as $\delta = d_{max}/M$, where $M$ is the downsampling factor, or according to the state of the DSN and network.  The DSN generates UTXOs pool ${\cal P}_{UTXO}$ from the reserved block bodies. When it receives a new transaction $T$, it will check $T$ with Algorithm \ref{alg:dtveri}.


%
%
%
%
%
%

\renewcommand\arraystretch{2}
\begin{table}[tbp]
\scriptsize
\centering
\caption{The Probability of Broadcast of DSNs}
\begin{tabular}{p{2.4cm}p{2.4cm}p{2.4cm}}
\hline
Probability &Broadcast &Discard\\ \hline
Valid &${\frac{N_{su}}{N_u}}$&${\frac{N_u-N_{su}}{N_u}}$\\ \hline
Invalid &${\frac{N_{st}-N_{su}}{N_t-N_u}}$&${\frac{N_t-N_u-(N_{st}-N_{su})}{N_t-N_u}}$\\ \hline
\end{tabular}
\end{table}
\renewcommand\arraystretch{1}

Let $N_t$ and $N_u$, respectively, denote the number of all the transaction outputs and the number of all the UTXOs. Let $N_{st}$ and $N_{su}$, respectively, denote the number of  transaction outputs and the number of UTXOs on a DSN. Table I gives the probability of broadcast of DSNs. If the transaction is valid, the DSN will broadcast it; otherwise discard. We define the broadcast accuracy as follows.

\begin{definition}
The broadcast accuracy $\varphi$ is the probability that a node broadcasts valid transactions.
\label{define:broadcast accuracy}
\end{definition}

For a DSN with the reserved set ${\cal D}\subseteq \{d_1,d_2,\ldots,d_\delta\}$, its broadcast accuracy is
\begin{equation}\label{eq_10}
{\varphi_{\cal D}=\frac{N_{su}}{N_u}}.
\end{equation}

We define $u(d)$ as the probability distribution of each block used by a new transaction. Then, the entropy and the broadcast accuracy could be expressed as
\begin{equation}\label{eq_81}
H(d)=E[-\log{u(d)}]=-u(d)\log{u(d)},
\end{equation}
\begin{equation}\label{eq_101}
{\varphi_{\cal D}=\frac{N_{su}}{N_u}}={\int_{\cal D} u(d)\mathrm{d}d}.
\end{equation}
Since $H(d)$ is proportional to $u(d)$, a higher sum of $H(d)$ means a higher sum of $u(d)$. Therefore, with the same number of reserved blocks, the entropy-based downsampling blockchain algorithm could achieve the highest broadcast accuracy among different ways of downsampling.

\subsection{Entropy of Bitcoin}
Here we use Bitcoin as an example to illustrate how to estimate the entropy of the blockchain $H(d)$. Since the new transaction is verified according to the UTXOs pool, the probability distribution of UTXOs could indicate the probability distribution of each block used by a new transaction. In other words, we could estimate $H(d)$ according to the distribution of UTXOs.

From the set theory, we could build the model of the blockchain based on UTXOs. Thus, a valid state can be seen as a set of UTXOs,
\begin{equation}
\boldsymbol{\sigma}_{t} \equiv \{txo^1_{t}, txo^2_{t}, txo^3_{t}, \dots, txo^n_{t}\},
\end{equation}
where $txo^i_{t}$ is the UTXO at the time slot $t$, ${i} \in \mathbb{N}$.

The inputs and outputs of a transaction can also be regarded as a set of transaction outputs in the UTXOs model,
\begin{equation}
T_{t}^{in} \equiv \{txo_{t}^{j_1}, txo_{t}^{j_2}, txo_{t}^{j_3}, \dots, txo_{t}^{j_{in}}\},
\end{equation}
\begin{equation}
T_{t}^{out} \equiv \{txo_{t+1}^{k_1}, txo_{t+1}^{k_2}, txo_{t+1}^{k_3}, \dots, txo_{t+1}^{k_{out}}\}.
\end{equation}
Thus, $T_{t}^{in} \subseteq \boldsymbol{\sigma}_{t}$ and $T_{t}^{out} \subseteq \boldsymbol{\sigma}_{t+1}$.

The state transition between two states can be seen as removing all $txos$ of transaction inputs from the previous state, and adding $txos$ of transaction outputs,
\begin{equation}
\boldsymbol{\sigma}_{t+1} \equiv \boldsymbol{\sigma}_{t} \bigtriangleup T.
\end{equation}

Thus, the most recent state of Bitcoin can be seen as the most recent set of UTXOs. However, the distribution of UTXOs changes as the block height increases, causing difficulty in analysis. Instead, we focus on state duration defined as follows.

\begin{definition}
State duration $x$ is
\begin{equation}\label{eq_surviva}
x \equiv d_{prod}-d_{used},
\end{equation}
where $d_{prod}$ and $d_{used}$ are the depth of the block where a UTXO was produced and used, respectively.
\end{definition}
Let us denote the distribution of $x$ by $N(x)$. For a stable blockchain, the distribution of the state duration is stable. In other words, the state duration is more universal. We can derive the distribution of UTXOs from the survival function of the state duration.

If every UTXO is random, independent and equally possible to use, there will be more UTXOs with shorter state duration than with longer state duration. In this situation, the state duration should conform to the exponential distribution.

The probability density function of the state duration is
\begin{equation}\label{eq_1}
f(x)=\frac{N(x)}{\int_0^{+\infty} N(x)\mathrm{d}x}.
\end{equation}
Its cumulative function is
\begin{equation}\label{eq_4}
C(d)=\int_0^{d-1} f(x)\mathrm{d}x.
\end{equation}
 For a block with depth $d$, each output is used with probability $C(d)$.

The probability of UTXOs in each block is
\begin{equation}\label{eq_5}
U(d)=1-C(d),
\end{equation}
which is the survival function of the state duration.

If the number of transaction outputs is similar, the probability density function $u(d)$ of UTXOs can be derived from $U(d)$,
\begin{equation}\label{eq_6}
u(d)=\frac{U(d)}{\int_0^{+\infty} U(d)\mathrm{d}d}.
\end{equation}

Based on the $224197$ blocks and $847656$ UTXOs of Bitcoin blockchain on April 21, 2018, we fit the distribution of the state duration to the function
\begin{equation}\label{eq_2}
N(x)=115000e^{-2.005x}+38850e^{-0.1302x},
\end{equation}
with $R-square=0.99$. We include $N(x)$ and the actual distribution of the state duration in Fig. \ref{fg:survivaltime}. It shows that $N(x)$ is close to the actual distribution. The state duration of Bitcoin is mainly distributed in smaller areas. Therefore, newer transaction outputs are more likely to be unspent.

\begin{figure}[htbp]
\centerline{\includegraphics[width=3.1in]{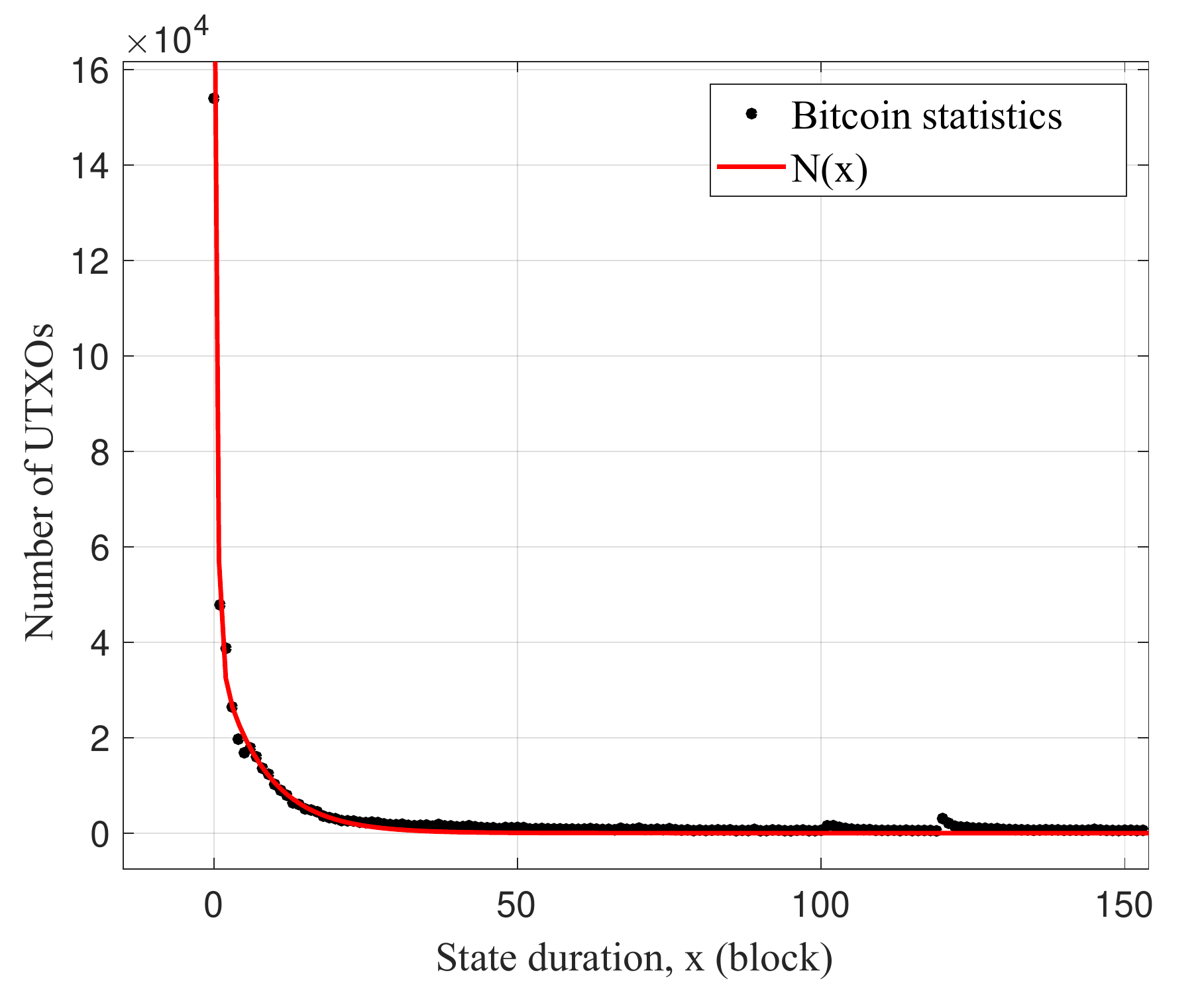}}
\caption{State duration of Bitcoin.}
\label{fg:survivaltime}
\end{figure}

According to the above $N(x)$, we have
\begin{equation}\label{eq_3}
f(x)=0.3233e^{-2.005x}+0.1092e^{-0.1302x},
\end{equation}
and
\begin{equation}\label{eq_7}
u(d)=0.0247e^{-2.005d}+0.1286e^{-0.1302d}.
\end{equation}

Then, we can give the entropy of Bitcoin \begin{equation}\label{eq_71}
\begin{array}{cc}
H(d)=&(0.0247e^{-2.005d}+0.1286e^{-0.1302d})\\
&\times log(0.0247e^{-2.005d}+0.1286e^{-0.1302d}).
\end{array}
\end{equation}

\subsection{Simulation results}
On the basis of theoretical analysis, we test the performance of entropy-based DSNs under various downsampling factors $M$.

\begin{figure}[htbp]
\centering
\includegraphics[width=3.1in]{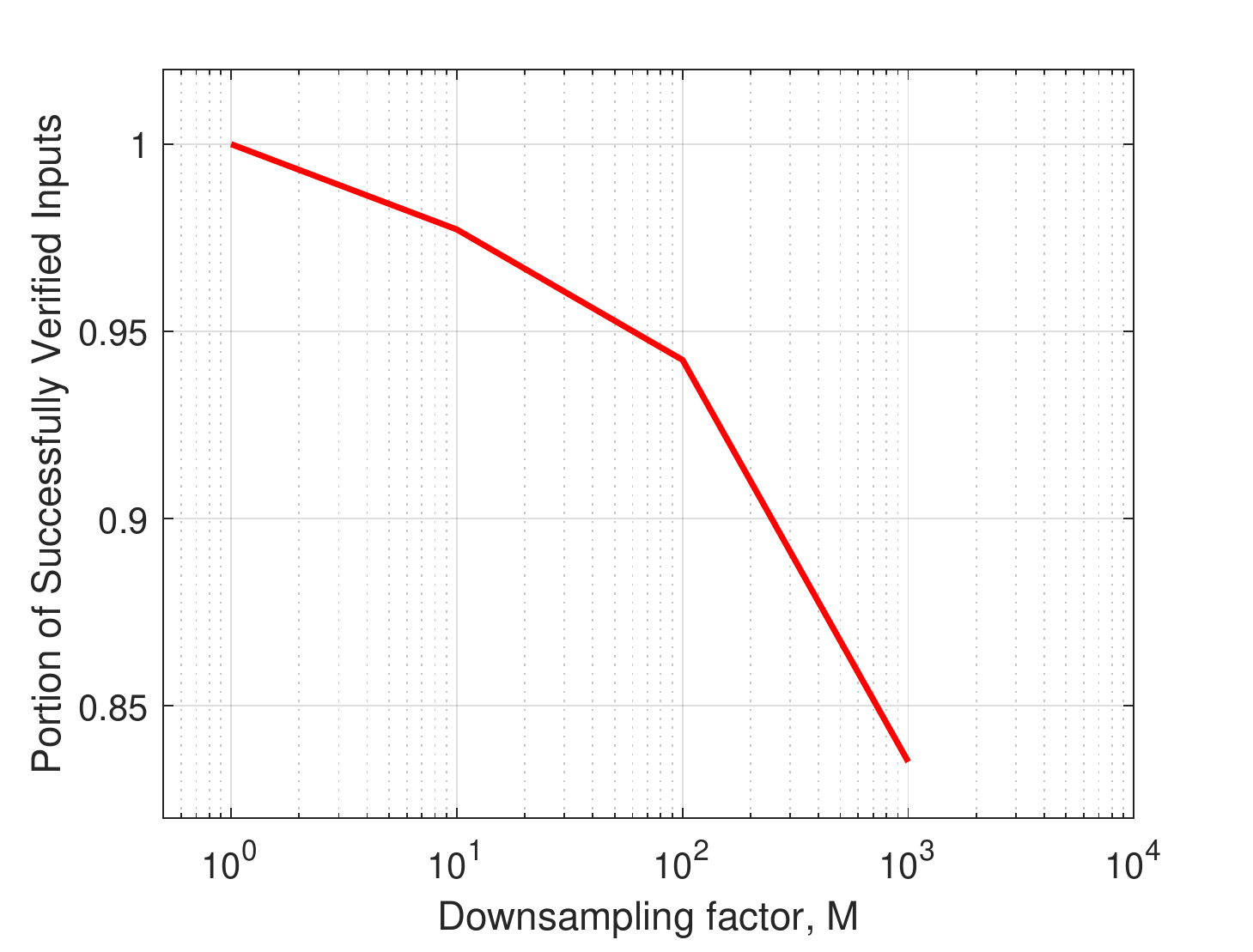}
\caption{Average broadcast accuracy of entropy-based DSNs under various downsampling factors.}
\label{fg:portionInputsAve}
\end{figure}

On June 17, 2019, we simulated entropy-based DSNs with the $6057$ transaction inputs of the Bitcoin. It can be seen from Fig.~\ref{fg:portionInputsAve} that even if the downsampling factor reaches $1000$, the DSN can obtain an average broadcast accuracy of $80\%$ or more. When the downsampling factor is $100$, the average broadcast accuracy of DSNs is over $90\%$. If the downsampling factor is $10$, the loss of average broadcast accuracy is only about $3\%$.

\begin{figure}[htbp]
\centering
\includegraphics[width=3.1in]{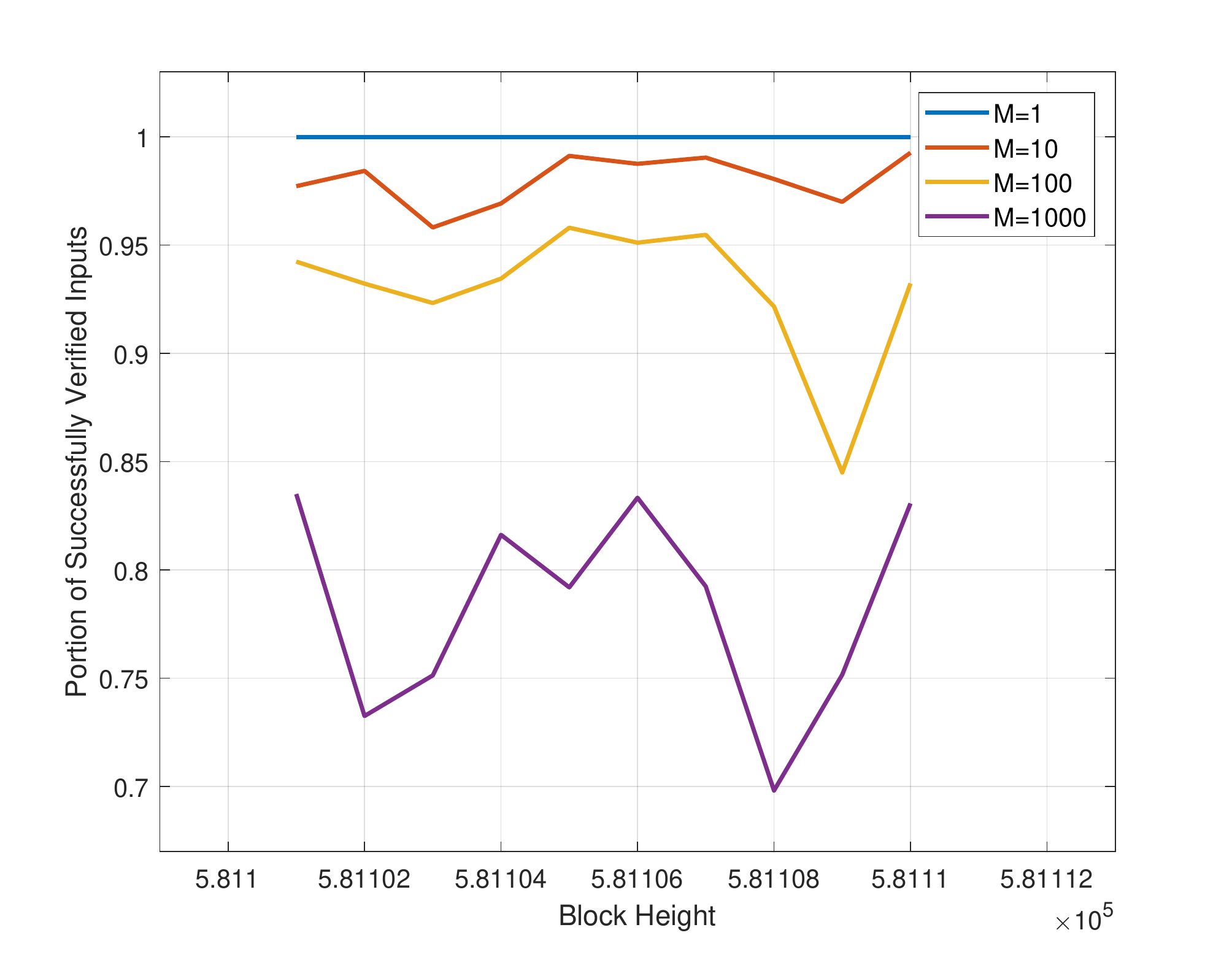}
\caption{Real-time broadcast accuracy of entropy-based DSNs under different downsampling factors.}
\label{fg:portionInputsNow}
\end{figure}

Then, we investigate the real-time broadcast accuracy for each block. We test $78$ blocks from block heights $581101$ to $581178$ of the Bitcoin. As show in Fig.~\ref{fg:portionInputsNow}, when processing the transactions of each block in real-time, even if the downsampling factor reaches $1000$, DSNs can obtain the real-time broadcast accuracy more than $70\%$. When the downsampling factor is $100$, the real-time broadcast accuracy of DSNs is over $80\%$. If the downsampling factor is $10$, the real-time broadcast accuracy of DSNs is over $95\%$.

In short, DSNs can significantly reduce the storage overhead, while keeping high broadcast accuracy both on average and in real-time.

\subsection{Security analysis}

Since a DSN stores all block headers and part of bodies, its security is always better than SPV nodes. Moreover, because of the sequential feature of blockchain, it can be determined whether a transaction output is UTXO if all blocks after that transaction output are known. An invalid transaction could not pass the verification process of DSNs, proved by the following theorems and lemma.

\begin{theorem}
If $T_{k+1}, T_{k+2}, T_{k+3}, \dots, T_{t}$ do not use a output of $T_{k}$ as the transaction input, this output is UTXO.
\end{theorem}

\begin{proof}
Assume that an output of $T_{k}$ has already been used.
\begin{equation}\label{proof_1}
  \because T_{i}^{in} \not= T_{k}^{out}, \forall i \in \{k+1, k+2, k+3, \dots, t\},
\end{equation}
as the transaction cannot use its own output as input,
\begin{equation}\label{proof_2}
  \therefore T_{i}^{in} = T_{k}^{out}, \exists i < k,
\end{equation}
which contradicts the time sequence of the blockchain system. Thus, the assumption is not true, and this output is UTXO.
\end{proof}

Similarly, it is easy to prove the following theorem.
\begin{theorem}
If $T_{k+1}, T_{k+2}, T_{k+3}, \dots, T_{t}$ use a output of $T_{k}$ as the transaction input, this output is not UTXO.
\end{theorem}

In addition, because transactions in blocks are continuous, it comes to the following lemma.
\begin{lemma}
If there are $B_{k+1}, B_{k+2}, B_{k+3}, \dots, B_{t}$, the generated set of UTXOs $\boldsymbol{\sigma} _{t}^{k+1} \subseteq \boldsymbol{\sigma} _{t}$.
\end{lemma}

Therefore, an invalid transaction can neither jeopardize the security of payment nor the security of broadcasting.

\section{Blockchain Recovery with DSNs}
In this section, we investigate the recovery of blockchain with DSNs. Besides downsampling in terms of blocks, we can downsample in terms of transactions. Since transactions are the smallest units of blockchain state transition, recovery in terms of transactions is flexible and efficient. Thus, here we further investigate recovery of blockchain in terms of transactions.

\subsection{DSNs with random transparent coding}
From the probability analysis \cite{1561992}, after $N$ balls are independently and randomly thrown into $K$ bins, and $K$ is large enough, the probability that all bins have a ball is $1-\epsilon$ when
\begin{equation}\label{eq_3}
N >K \log _{e} \frac{K}{\epsilon}.
\end{equation}

In other words, $K$ transactions could be recovered with probability $1-\epsilon$, if $K \log _{e} \frac{K}{\epsilon}$ transactions are randomly stored. Usually, we use erasure coding to ensure data recovery. However, parity symbols of erasure coding can not be verified by hash in the blockchain. It brings serious issues in security and availability.

In order to make data on each node directly be verified and used, we propose to transparently store data, and only encode during data recovery as follows.

The storage algorithm of a node $P$, which is the downsampling algorithm with transparent coding, can be described as Algorithm \ref{alg:1}. This algorithm constructs a spatial distribution of transactions on the blockchain network.
\begin{algorithm}[htb]
  \caption{ Downsampling with transparent coding}
  \label{alg:1}
  \begin{algorithmic}[1]
    \State // $P$ stores the set of entire block headers $\mathcal{H}$ and the set of the identifier (hash value) $\mathcal{I}$ of every transaction.
    \State $P$: STORE($\mathcal{H}, \mathcal{I}$);
    \State // P generates a positive integer random variable $\gamma$ with a mean of $\log _{e} \frac{K}{\epsilon}$.
    \State $\gamma =$ RANDOMI($1,K, \log _{e} \frac{K}{\epsilon}$);
    \State $P$: uniformly and randomly select $\gamma$ distinct identifiers $I_1, I_2,...,I_{\gamma} \in \mathcal{I}$;
    \For{j=1 to $\gamma$}
      \State $P$: get transaction $T_j$, where the identifier of $T_j$ is $I_j$;
      \State $P$: STORE($T_j$);
    \EndFor
  \end{algorithmic}
\end{algorithm}

When a node $P'$ needs a transaction $T'$, it can request the transaction via Algorithm \ref{alg:2}. The probability that the transaction $T'$ cannot be found is $(1-\log _{e} \frac{K}{\epsilon})^{\cal{M}}$, if all neighbor nodes are encoded with $\bar{\gamma} = \log _{e} \frac{K}{\epsilon}$, where $\cal{M}$ is the number of neighbor nodes and $\bar{\gamma}$ is the average of $\gamma$.

\begin{algorithm}[htb]
  \caption{ Recovery of a transaction}
  \label{alg:2}
  \begin{algorithmic}[1]
    \Require
      \Statex The set of neighbor nodes of node $P'$, ${\{P^1,P^2,\dots,P^{\cal{M}}\}}$;
      \Statex The identifier of the transaction $T'$, ${I}_{T'}$;
    \Ensure
      The transaction $T'$;
    \For{i=1 to $\cal{M}$}
      \State //$P'$ sends ${I}_{T'}$ to $P^i$.
      \State $P' \to P^i$: ${I}_{T'}$;
      \If {$P^i \to P'$: $T'$} \Return $T'$
      \EndIf
    \EndFor
  \end{algorithmic}
\end{algorithm}

In addition, when a node $P'$ needs to recover all transactions of a block, although Algorithm \ref{alg:2} can be used in parallel, it will waste bandwidth, meaning that $O(K \ln (k / \epsilon))$ transactions need to be propagated. This problem can be solved by transmitting the bitwise sum, modulo $2$, of ${\gamma}^i$ transactions of each node. The codeword sent by node $P^i$ is
\begin{equation}\label{eq_4}
  c^i = T_1 \oplus T_2 \oplus T_3 \oplus \dots \oplus T_{{\gamma}^i},
\end{equation}
where ${\gamma}^i$ is the degree of the codeword. Since the average degree $\bar{\gamma}$ is significantly less than $K$, we can decode these codewords like decoding a sparse-graph code over an erasure channel. Here we use the message passing algorithm, which is described in detail in Algorithm \ref{alg:3}.

\begin{algorithm}[htb]
  \caption{ Decoding of transparent coding}
  \label{alg:3}
  \begin{algorithmic}[1]
    \State Get G; // Get the coding matrix.
    \Repeat
      \State find a codeword $c^i$ that is connected to only one transaction $T_j$;
      \If {$c^i == null$}
        \Return false;
      \EndIf
      \State $T_j = c^i$;
      \ForAll {$i'$ such that $G_{i'j} == 1$}
      \State $c^{i'} = c^{i'} \oplus T_j$;
      \State $G_{i'j} = 0$
    \EndFor
    \Until{all $T_j$ are determined}
  \end{algorithmic}
\end{algorithm}

\subsection{DSNs with robust soliton distribution transparent coding}
Although coding can reduce bandwidth consumption, in practice the random selection of degrees performs poorly because it is very likely that there is no codeword with the degree-one at some point in the decoding. So the choice of the degree distribution is very important. Therefore, inspired by fountain codes \cite{1181950}, we find that the robust soliton distribution can be used as the probability distribution of the random variable in Algorithm \ref{alg:1}, which ensures that the expected number of codewords with the degree-one is around
\begin{equation}
S \equiv c \log _{e}(K / \epsilon) \sqrt{K},
\end{equation}
where $\epsilon$ is the probability that the node $P'$ could not decode all transactions after $K'=KZ$ codewords have been received, and $c$ is a constant of order $1$. The robust soliton distribution is
\begin{equation}
\mu(\gamma)=\frac{\rho(\gamma)+\tau(\gamma)}{Z},
\end{equation}
where
\begin{equation}
\begin{array}{l}{\rho(1)=1 / K} \\ {\rho(\gamma)=\frac{1}{\gamma(\gamma-1)} \quad \text { for } \quad \gamma=2,3, \ldots, K},\end{array}
\end{equation}

\begin{equation}
\tau(\gamma)=\left\{\begin{array}{ll}{\frac{S}{K} \frac{1}{\gamma}} & {\text { for } \gamma=1,2, \ldots,(K / S)-1} \\ {\frac{S}{K} \log (S / \epsilon)} & {\text { for } \gamma=K / S} \\ {0} & {\text { for } \gamma>K / S}\end{array}\right.,
\end{equation}
and
\begin{equation}
Z=\Sigma_{\gamma} (\rho(\gamma)+\tau(\gamma)).
\end{equation}

\begin{figure}[tp]
\centering
\includegraphics[width=3.3in]{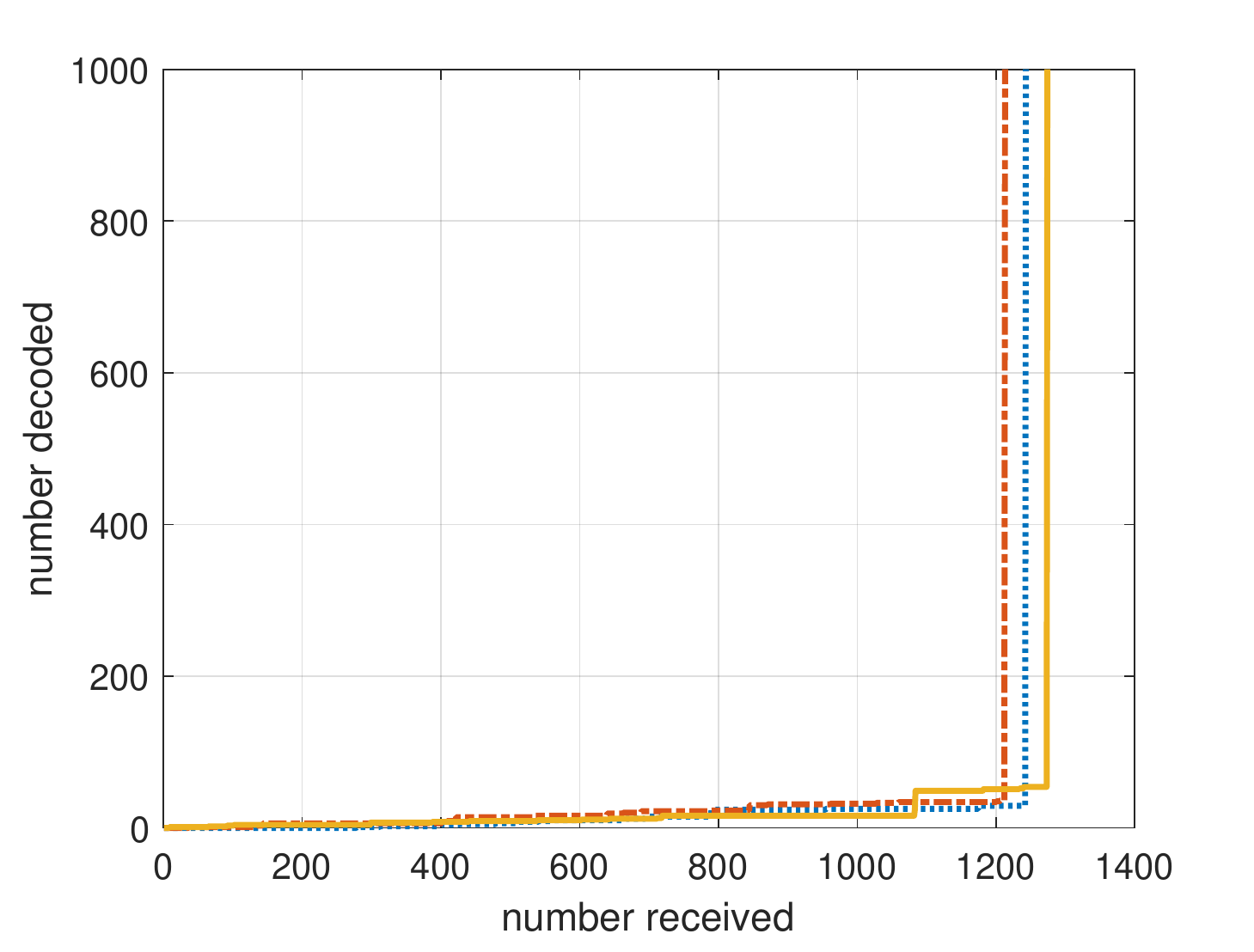}
\caption{Performance of robust soliton distribution transparent blockchain codes.}
\label{fig:1}
\end{figure}
In this case, the decoding algorithm of transparent blockchain codes is the same as the decoding algorithm of fountain codes. Therefore, any $K+O\left(\sqrt{K} \ln ^{2}(K / \epsilon)\right)$ codewords from different nodes can recover the $K$ original transactions with probability $1-\epsilon$, and each node only needs to store $O(\ln (K / \epsilon))$ transactions on average. Figure \ref{fig:1} shows three different decoding runs of transparent blockchain codes, where $c=0.05, \epsilon = 0.05, K=1000, \bar{\gamma} = 10.40$. The original block could be recovered when the number of received codewords is around 1300.

\subsection{Security analysis}
DSNs with transparent coding have two following properties to meet the requirement of safety:
\begin{itemize}
  \item {\bf Anti-fraud:} it is computationally infeasible to defraud a DSN to trust a forged transaction.
  \item {\bf Anti-obstruction:} it is computationally infeasible to have a DSN spend much more than normal to verify a transaction.
\end{itemize}
These two properties could be formalized in terms of two games that we play with two adversaries, respectively.

In the anti-fraud game, there is an adversary who claims that he can defraud a DSN to trust a forged transaction and a challenger that will test this claim. Here the DSN is the challenger. We are going to allow the adversary to run the hash function of his choice inputs, for as long as he wants, as long as the number of guesses is plausible. Once the adversary is satisfied that he has tried enough inputs, then the adversary picks a transaction and attempts to forge a transaction. However, since the hash function is collision resistant, the challenger will detect the inconsistency between the forged transaction and the hash identifier. Thus, the challenger will win the game.

In addition, in the anti-congestion game, there is an adversary who claims that he can have a DSN spend much more than normal to verify a transaction and a challenger that will test this claim. Here the DSN is the challenger. We are going to allow the adversary to send any codeword to the challenger, for as long as the number of codewords is plausible. The adversary attempts to obstruct the challenger to verify transactions. However, since the maximum degree of the codeword is $K$, the challenger will be able to decode any codeword within $K-1$ bitwise sums, modulo $2$. Therefore, the challenger will win the game.
\section{Conclusion}
In this paper, we proposed to downsample history data to reduce the storage overhead of blockchain. It demonstrates that entropy based downsampling can provide high verification and broadcast accuracy. Moreover, in order to recover data in a practical way, we proposed to transparently store data following robust soliton distribution, and only encode during data recovery. It demonstrates that our proposed transparent coding provides reliability with low bandwidth consumption. In short, our transparent downsampling blockchain algorithm has good scalability in independency and recovery.

\bibliographystyle{IEEEtranTCOM}
\bibliography{IEEEabrv,reference}
\end{document}